\documentclass[conference]{IEEEtran}
\IEEEoverridecommandlockouts
\usepackage{cite}
\usepackage{graphicx}
\usepackage{textcomp}
\usepackage{xcolor}

\usepackage{mathpazo}
\usepackage[latin1]{inputenc}
\usepackage{babel}
\usepackage{algcompatible}
\usepackage{algorithm}
\usepackage{algpseudocode,algorithmicx}

\floatname{algorithm}{Algorithm}
\algnewcommand{\algorithmicand}{\textbf{ and }}
\algnewcommand{\algorithmicor}{\textbf{ or }}
\algnewcommand{\algorithmicnot}{\textbf{ not }}
\algnewcommand{\algorithmicfalse}{\textbf{ false }}
\algnewcommand{\algorithmictrue}{\textbf{ true }}
\algnewcommand{\algorithmicbreak}{\textbf{break}}
\algnewcommand{\OR}{\algorithmicor}
\algnewcommand{\AND}{\algorithmicand}
\algnewcommand{\NOT}{\algorithmicnot}
\algnewcommand{\FALSE}{\algorithmicfalse}
\algnewcommand{\TRUE}{\algorithmictrue}
\algnewcommand{\Break}{\algorithmicbreak}

\usepackage{balance}
\usepackage{amsmath}
\usepackage{amssymb}
\usepackage[hashEnumerators, smartEllipses]{markdown}
\usepackage{amsthm}

\usepackage{parskip}
\usepackage[T1]{fontenc}
\usepackage{amsfonts}
\usepackage{textcomp}
\usepackage{bm}
\usepackage[all]{xy}
\usepackage{pifont}
\usepackage{hyperref}
\hypersetup{
	colorlinks=true,
	linkcolor=black,
	filecolor=black,      
	urlcolor=black,
	citecolor=black,
}

\usepackage{listings}
\usepackage{xcolor}
\usepackage{sistyle} 
\newcommand\bSI[1]{{\small\SI{}{#1}}}

\makeatletter
\newlength\unitwdth
\newlength\numwdth
\newlength\tdima
\newcommand\SIdescr[3]{%
	\setlength\tdima{\linewidth}%
	\addtolength\tdima{\@totalleftmargin}%
	\addtolength\tdima{-\dimen\@curtab}%
	\addtolength\tdima{-\unitwdth}%
	\addtolength\tdima{-\numwdth}%
	\parbox[t]{\tdima}{%
		#1
		\leaders\hbox{$\m@th\mkern \@dotsep mu\hbox{\tiny}\mkern \@dotsep mu$}%
		\hfill
		\ifhmode\strut\fi
		\makebox[0pt][l]{%
			\makebox[\unitwdth][l]{\quad\bSI{#2}}%
			\makebox[\numwdth][l]{\quad #3}}}}
\makeatother

\definecolor{codegreen}{rgb}{0,0.6,0}
\definecolor{codegray}{rgb}{0.5,0.5,0.5}
\definecolor{codepurple}{rgb}{0.58,0,0.82}
\definecolor{backcolour}{rgb}{0.95,0.95,0.92}

\lstdefinestyle{mystyle}{
	backgroundcolor=\color{backcolour},   
	commentstyle=\color{codegreen},
	keywordstyle=\color{magenta},
	numberstyle=\tiny\color{codegray},
	stringstyle=\color{codepurple},
	basicstyle=\ttfamily\footnotesize,
	breakatwhitespace=false,         
	breaklines=true,                 
	captionpos=b,                    
	keepspaces=true,                 
	numbers=left,                    
	numbersep=5pt,                  
	showspaces=false,                
	showstringspaces=false,
	showtabs=false,                  
	tabsize=2
}

\usepackage{usnomencl}
\usepackage[numbers,sort&compress]{usbib}
\addto{\captionsafrikaans}{}
\addto{\captionsenglish}{}

\newenvironment{example}[1][Example.]{\begin{trivlist}\item[\hskip \labelsep {\bfseries #1}]}{\end{trivlist}}

\usepackage{graphicx}
\usepackage{color}
\usepackage{eso-pic}

\usepackage{calc} 
\usepackage{float}
\newcommand{\compresslist}{ 
	\setlength{\itemsep}{1pt}
	\setlength{\parskip}{0pt}
	\setlength{\parsep}{0pt}
}

\algnewcommand\algorithmicclass{\textbf{class}}
\algdef{SE}[CLASS]{Class}{EndClass}[1]{\algorithmicclass\ \textproc{#1}}{\algorithmicend\ \algorithmicclass}

\algnewcommand\algorithmicenum{\textbf{enum}}
\algdef{SE}[ENUM]{Enum}{EndEnum}[1]{\algorithmicenum\ \textproc{#1}}{\algorithmicend\ \algorithmicenum}

\algnewcommand\algorithmicwhen{\textbf{when}}
\algdef{SE}[WHEN]{When}{EndWhen}[1]{\algorithmicwhen\ #1 \algorithmicdo}{\algorithmicend\ \algorithmicwhen}

\algnewcommand\algorithmicconstructor{\textbf{constructor}}
\algdef{SE}[CONSTRUCTOR]{Constructor}{EndConstructor}[1]{\algorithmicconstructor\ (#1)}{\algorithmicend\ \algorithmicconstructor}
\usepackage{tikz}
\usetikzlibrary{shapes}
\usetikzlibrary{decorations.pathreplacing}
\usetikzlibrary{calc}
\usetikzlibrary{arrows}

\tikzstyle{decision} = [diamond, draw, fill=gray!20, 
text width=4.5em, text badly centered, node distance=3cm, inner sep=0pt]
\tikzstyle{block} = [rectangle, draw, fill=blue!20, text width=5em, text centered, rounded corners, minimum height=1.5em]
\tikzstyle{line} = [draw, -latex']
\tikzstyle{cloud} = [draw, ellipse,fill=red!20, node distance=3cm, minimum height=2em]
\usetikzlibrary{decorations.markings}
\tikzstyle{vertex}=[circle, draw, inner sep=0pt, minimum size=0pt]

\usepackage{caption} 
\usepackage{subcaption} 
\usepackage{color}
\usepackage{eso-pic}%
\usepackage{xcolor}
\usepackage{multirow}
\usepackage{bookmark}
\theoremstyle{plain}
\newtheorem{thm}{{Theorem}}[section]
\newtheorem{lem}[thm]{Lemma}

\newtheorem{defn}[thm]{Definition}

\def\BibTeX{{\rm B\kern-.05em{\sc i\kern-.025em b}\kern-.08em
    T\kern-.1667em\lower.7ex\hbox{E}\kern-.125emX}}
\begin{document}

\title{
	Distributed Identification of Central Nodes with Less Communication
\thanks{JFM thanks the AIMS Research Centre and the DST-CSIR HCD Inter-Bursary
	scheme for PhD funding.}
}

\author{\IEEEauthorblockN{1\textsuperscript{st} Jordan F.~Masakuna}
\IEEEauthorblockA{\textit{Computer Science Division} \\
\textit{Stellenbosch University}, Stellenbosch\\
\textit{Impact Radius}, Cape Town\\
South Africa \\
jordan@aims.ac.za}
\and
\IEEEauthorblockN{2\textsuperscript{nd} Steve Kroon}
\IEEEauthorblockA{\textit{Computer Science Division} \\
	\textit{Stellenbosch University}, Stellenbosch\\
	South Africa \\
	kroon@sun.ac.za}
}

\maketitle

\begin{abstract}

This  paper  is  concerned  with  distributed detection of central nodes in complex networks using closeness centrality. 
Closeness centrality plays an essential role in network analysis.  Evaluating closeness centrality exactly requires complete knowledge of the network; for large networks, this may be inefficient, so closeness centrality should be approximated. Distributed tasks such as leader election can make effective use of centrality information for highly central nodes, but complete network information is not locally available. 
This paper refines a distributed centrality computation algorithm by You et al. \cite{you2017distributed} by pruning nodes which are almost certainly not most central. For example, in a large network, leave nodes can not play a central role. 
This leads to a reduction in the number of messages exchanged to determine the centrality of the remaining nodes. 
Our results show that our approach reduces the number of messages for networks which contain many prunable nodes. Our results also show that reducing the number of messages may have a positive impact on running time and memory size.
\end{abstract}


\section{Introduction}\label{sec:introduction}

Centrality metrics play an essential role in network analysis \cite{lam1979congestion}. For some types of centrality metrics such as betweenness centrality, evaluating network centrality exactly requires complete
knowledge of the network; for large networks, this may be too costly computationally, so approximate methods (i.e. methods for building a view of a network to compute centrality) have been proposed. Centrality information for highly central nodes can be used effectively for distributed tasks such as leader election, but distributed nodes do not have complete network information. 
The relative centrality of nodes
is important in problems such as selection of informative
nodes, for example, in active sensing \cite{nelson2006sensory}\textemdash i.e. the propagation time required to synchronize the nodes of a network can be minimized if the most central node is
known \cite{ramirez1979distributed, kim2013leader}. Here we consider closeness centrality\textemdash We wish to detect
most central nodes in the network. 

Not all nodes ultimately need to build a view of the network formed in their interaction since some nodes may realize quickly that they are not suitably central, i.e. they have small closeness centralities. The question is how to identify such insignificant nodes in a decentralised algorithm?
This paper tackles this problem by introducing a pruning strategy which reduces the number of messages exchanged between nodes compared to the algorithm from \cite{you2017distributed}. When a leader is chosen based on closeness centrality, we observe that in the algorithm of \cite{you2017distributed}, even nodes which can not play a central role, for example leaves, overload communication by receiving messages.\footnote{Except in special cases, a leaf node should not play a central role.}

%
%
This work proposes modifications to You et al's  decentralised method to construct a view of a communication graph for distributed computation of node closeness centrality. 
Since we consider an approximate and decentralised method to construct a view of a communication graph, inevitably nodes will construct different topologies describing their interaction network. We will use the term view as shorthand for view of a communication network. 

Our proposed method can be applied to arbitrary distributed networks, and is most likely to be valuable when  nodes form very large networks: reducing the number of messages will be a more pressing concern in large-scale networks. Such applications include  instrumented cars, monitoring systems, mobile sensor networks or general mobile ad-hoc networks. There are many reasons for reducing the number of messages in distributed systems. Here, we consider the following reason. A careful treatment of network communication of agents under weak signal conditions is crucial
\cite{tomic2012toward}.

\textbf{Contributions.}
\label{sec:network_contribution}
%
%
At  each iteration of view construction, each node prunes some nodes in its neighbourhood once and thereafter interacts only with the unpruned nodes to construct its view. We refer to this approach as pruning. The more of these prunable nodes a network contains the better our  algorithm performs relative to the algorithm in \cite{you2017distributed} in terms of number of messages. 
%
We empirically evaluate our approach on a number of benchmark networks from  \cite{leskovec2005graphs} and \cite{leskovec2007graph}, as well as some randomly generated networks, and observe our method outperforms the benchmark method \cite{you2017distributed} in terms of number of messages exchanged during interaction. We also observe some positive impact that pruning has on running time and memory usage.


We make the following assumptions as in \cite{you2017distributed}: 
\begin{itemize}
	\compresslist
	\item nodes are uniquely identifiable;
	\item a node knows identifiers of its neighbours;
	\item communication is bidirectional, FIFO and asynchronous;
	\item each agent is equipped with its own round counter.
\end{itemize} 

The rest of the paper is organised as follows:
Section \ref{sec:related-work} discusses the state of the art for decentralised computation of closeness centrality distribution of networks.
The new algorithm will be discussed in Section \ref{sec:description}.
Section \ref{sec:results} discusses the results. 
In Section \ref{sec:conclusion}, we conclude and propose further work.

\section{Background and related work}
\label{sec:related-work}

In a distributed system, the process of building a view of the network can be centralised or decentralised. It is centralised when the construction of the view of a network is performed by a single node, known as an initiator. Once the initiator has the view of a network, it may send it to the other nodes \cite{naz2017distributed}. If the initiator is not known in advance, the first stage of constructing a view of a network will involve a selection of the initiator. 

Decentralised methods to construct views can be exact or approximate.  Recent literature on methods to construct views can be found in \cite{naz2017distributed}.
%
%
%
%
%
%
%
Exhaustive methods build the complete topology of the communication graph\textemdash they involve an all-pairs shortest path computation. As a consequence, exact approaches suffer from problems of scalability \cite{naz2017distributed}. This can be a particular issue for large networks of nodes with constrained computational power and restricted memory resources.

To overcome the problem mentioned above that exhaustive methods suffer from, approximate methods have been proposed. Unlike exhaustive methods, approximate methods do not result complete knowledge of the communication graph. 


Many distributed methods for view construction are centralised, i.e. they require a single initiator in the process of view construction. 
The disadvantage of approximate methods is that the structure of a view depends on the choice of the initiator. 
An interesting distributed approach for view construction can be found in \cite{kim2013leader}, 
%
%
where an initiator constructs a tree as the view.

%

To the best of our knowledge and according to \cite{naz2017distributed}, decentralised approximate methods for view construction are very scarce because many methods for view construction assume some prior information about the network, so centralised methods are more appropriate. The method proposed by You et al. \cite{you2017distributed} seems to be the state of the art for decentralised approximate methods for view construction, i.e. views are constructed only from local interactions. This method simply runs  breadth-first search \cite{skiena1998algorithm} on each node. 

We next show the decentralised construction of a view using the algorithm in \cite{you2017distributed}.
\subsection{Decentralised view construction}
\label{sec:leaderelection_revue}
%
%

%
We consider the method proposed in You et al. \cite{you2017distributed}\textemdash the ``YTQ method''\textemdash as the state of the art for decentralised construction of a view of a network. The YTQ method was proposed for decentralised approximate computation of centrality measures (closeness, degree and betweenness centralities) of nodes in arbitrary networks. As treated in \cite{you2017distributed}, these computations require a limited view.  At the end of the interaction between nodes, each node can estimate its centrality based on its own view. 
In the following, we show how nodes  construct views using the YTQ method. Here we consider connected and unweighted graphs, and we are interested in computation of closeness centrality. 

Let $\delta_{ij}$ denote the path distance between the nodes $v_i$ and $v_j$ in an (unweighted) graph $G$\label{you:G} with vertex set $\mathcal{V}$ and edge set $\mathcal{E}$. The path distance between two nodes is the length of the shortest path between these nodes. 
\begin{defn}
	The closeness centrality \cite{bavelas1950communication} of a node is the reciprocal of the average path distance from the node to all other nodes. 
	Mathematically, the closeness centrality $c_i$\label{you:ci} is given by
	\begin{equation}
	\label{eq:closeness}
	c_i =  \frac{|\mathcal{V}|-1}{\sum_{j} \delta_{ij}}\,.
	\end{equation}
\end{defn}	
Nodes with high closeness centrality score have short average path distances to all other nodes.

 Each node's view is gradually constructed based on message passing. 
Each node sends its neighbour information to all of its immediate neighbours which relay it onward through the network. Communication between nodes is asynchronous, i.e. there is no common clock signal between the sender and receiver. 
\begin{algorithm}
	\begin{algorithmic}[1]
		\scriptsize
		\Procedure{runDistributedYTQ($\mathcal{N}_i, D$)}{}
		\State \texttt{YTQObject}$\gets $\textsc{distributedYTQ}($\mathcal{N}_i, D$)
		\While{\NOT \texttt{YTQObject}.\textsc{isEnded()}}
		\State \texttt{YTQObject}.\textsc{oneHop()}
		\State \texttt{YTQObject}.\textsc{update()}
		\EndWhile

		\EndProcedure
		\State
		\Class{distributedYTQ}
		\State \textbf{\underline{Class variables}}
		\State \makebox[2cm][l]{$D$} the pre-set maximum number of iterations
		\State  \makebox[2cm][l]{$\mathcal{M}_i$} a queue  used for messages received by the node $v_i$
		\State \makebox[2cm][l]{$\mathcal{N}_i$} a list of immediate neighbours of the node $v_i$
		\State \makebox[2cm][l]{$\mathcal{N}_i^{(t)}$} the set of $(t+1)$-hop neighbours of $v_i$
		\State \makebox[2cm][l]{$\mathcal{N}_{i, t}$} the set nodes known by $v_i$ until the end of iteration $t$
		\State \makebox[2cm][l]{$c_i$} the closeness centrality of node $v_i$
		\State \makebox[2cm][l]{$\delta_i$} a variable used for computation of $c_i$
		\State \makebox[2cm][l]{$t$} the current iteration number
		\State \makebox[2cm][l]{$T$} the actual maximum number of iterations
		\\\hrulefill
		\Constructor{$\mathcal{N}_i, D$}
		\State $(t, D, \delta_i) \gets (0, D, |\mathcal{N}_i|)$
		\State $\mathcal{N}^{(0)}_i\gets  \{v_j:\,v_j\in\mathcal{N}_i \}$
		\State $\mathcal{N}_{t, i}\gets \mathcal{N}^{(0)}_i$.\texttt{clone()}
		\EndConstructor
		\State
		\Procedure{oneHop()}{}
		\If{\NOT \textsc{isEnded()}}
		\For{$v_j \in\mathcal{N}_i$}
		\State $v_i$ sends $\langle \texttt{NeighbouringMessage($i, \mathcal{N}_i^{(t-1)}$)}\rangle$
		to $v_j$
		\EndFor
		\Else
		\State $T\gets \min(t, D)$
				\State $c_i\gets $ \textsc{closenessCentrality}{($T$)}
		\EndIf
		\EndProcedure
		\State
		\Procedure{update()}{}
		\If{\NOT \textsc{isEnded()}}
		\State $t\gets t+1$
		\State $\mathcal{N}^{(t)}_i\gets \emptyset$
		\While{$\mathcal{M}_i.\texttt{size()}\geq 1$}
		\State $(j, \mathcal{N}_j^{(t-1)}) \gets \mathcal{M}_i.\texttt{dequeue}()$
		\State $\mathcal{N}^{(t)}_i\gets \mathcal{N}^{(t)}_i \cup \mathcal{N}^{(t-1)}_j$
		\Comment{``$v_i$ fuses the messages received''}
		\EndWhile
		\State $\mathcal{N}^{(t)}_i\gets \mathcal{N}^{(t)}_i\setminus \mathcal{N}_{t-1, i}$
		\State $\mathcal{N}_{t, i}\gets \mathcal{N}_{t-1,i}\cup \mathcal{N}^{(t)}_i$
		\State $\delta_i\gets \delta_i +t|\mathcal{N}^{(t)}_{i}|$
		
		\EndIf
		\EndProcedure
		\State

		\Function{closenessCentrality}{$T$}
		\State \Return $ \frac{|\mathcal{N}_{i, T}|-1}{\delta_i}$
		\EndFunction
		\State
		\Function{isEnded()}{}
		
		\State \Return $\mathcal{N}_i^{(t)}=\emptyset$ \OR $t=D$
		\label{line:youend}
		\EndFunction
		
		\EndClass
	\end{algorithmic}  
	\caption[\textit{Our presentation of the process given in \cite{you2017distributed}}.]{\textit{Our presentation of the YTQ method in \cite{you2017distributed}. The algorithm gives the code run on a single node $v_i$. 
			} \textit{An example of pseudocode is given in the procedure \textsc{runDistributedBFS}.}}
	\label{algo:construction0}
\end{algorithm}
The YTQ method \cite{you2017distributed} that each node $v_i$ uses to construct its view is given in Algorithm \ref{algo:construction0}. 

Let $\mathcal{N}_i$\label{you:Ni} be the set of neighbours of $v_i$ \label{you:vi} and $\mathcal{N}_i^{(t)}$\label{you:Nt} the set of nodes at distance $t+1$ from a node $v_i$, so $\mathcal{N}^{(0)}_i =\mathcal{N}_i$. The initial set of neighbours, $\mathcal{N}_i$, is assumed to be known. 

During each iteration, each node sends its neighbourhood information to all its immediate neighbours. We are restricted to peer-to-peer communication because nodes have limited communication capacity. Each node waits for communication from all of its direct neighbours after which it updates an internal round counter.  
A node $v_i$ stores messages received in a queue, represented by $\mathcal{M}_i$. 
After round $t\geq 1$, the topology of $v_i$'s view  is updated as follows
\begin{equation}
\label{you:NiD}
\mathcal{N}^{(t)}_i\gets \bigcup_{v_j\in \mathcal{N}_i}\mathcal{N}^{(t-1)}_j \setminus \mathcal{N}_{i, t}\,,
\end{equation}
where
\begin{equation}
\mathcal{N}_{i, t}\gets  \bigcup_{k=0}^{t-1}\mathcal{N}^{(k)}_i\,.
\end{equation}

The algorithm terminates after at most $D$\label{you:D} iterations, where $D$ is an input of the algorithm.\footnote{It can also be set or determined in a distributed manner (i.e. the value of $D$ can be determined by nodes during interaction) as in \cite{garin2012distributed}.} In this paper, we consider a pre-set value of $D$. However, some nodes can also reach their equilibrium stage before the iteration $D$ (e.g. nodes which are more central than others). Such nodes need to terminate when equilibrium is reached (see Line \ref{line:youend} in Algorithm \ref{algo:construction0}). A node $v_i$ reaches equilibrium at iteration $t$ when
$$
\mathcal{N}^{(t)}_i = \emptyset\,.
$$
At the end, every node has a view of network, and so the required centralities can be calculated locally. This view construction method is approximate when the total number of iterations $D$ is less than the diameter of the graph, otherwise the method is exhaustive, i.e. all views correspond to the exact correct information, assuming a failure-free scenario. With a decentralised approximate method, nodes may have different views at the end. Each node will evaluate its closeness centrality based on its own view of the network.

\section{Decentralised view construction}
\label{sec:description}

The idea behind pruning technique is that, during view construction, some nodes can be pruned (i.e. some nodes will stop relaying neighbour information). Pruned nodes are not involved in subsequent steps of the algorithm and their closeness centralities are treated as zero. 
%

%
Our approach thus applies pruning after each iteration $t$ of communication of the YTQ method. During the pruning stage, each node checks whether it or any nodes in its one-hop neighbourhood should be pruned. 
Nodes can identify the other nodes in their neighbourhood being pruned so that they do not need to wait for or send messages to them in subsequent iterations. This reduces the number of messages exchanged between nodes.

Given two direct neighbours, their sets of pruned nodes are not necessarily the same, and nodes do not need to exchange such information between themselves. 
\subsection{Pruning}
Before describing our proposed pruning method, we argue that pruning preserves information of most central nodes of a graph using closeness centrality.
\subsubsection{Theoretical justification}
While we are aiming to estimate closeness centrality distribution on a graph using pruning, the concept of pruning can directly be related to eccentricity centrality \cite{hage1995eccentricity}. The eccentricity of a node is the maximum distance between the node and another node. Eccentricity and eccentricity centrality are reciprocal to each other. We consider the following points to achieve our goal.
\begin{itemize}
	\compresslist
	\item Pruned nodes have relatively high eccentricities (as will be discussed later in Lemma \ref{lem:eccentricity}).
	\item Previous studies show that eccentricity and closeness centralities are strongly positively correlated for various types of graphs \cite{batool2014towards, meghanathan2015correlation}. This is partly due to the fact that they both operate on the concepts of paths.
\end{itemize}
From what precedes, our proposed pruning method is then recommended for approximations of closeness centralities for categories of graphs where eccentricity and closeness centralities are highly correlated. 
\subsubsection{Description}
We will first show how a node identifies prunable nodes after each iteration.
There are two types of objects (leaves, and nodes causing triangles) that a node can prune after the first iteration.  When a node is found to be of one of these types, it is pruned.
Since nodes have learnt about their $2$-hop neighbours after the end of the first iteration, a node $v_i$ knows the neighbours $\mathcal{N}_j$ of each of its direct neighbours $v_j$. Our pruning method is decentralised, so each node is responsible to identify prunable nodes in its neighbourhood, including itself. Let $d_i$ denote the degree of node $v_i$.
\begin{defn}[A node causing a triangle]
	\label{defn:node_triangle}
A non-leaf node $v_j$ causes a triangle if $d_j=2$ and its two immediate neighbours are immediate neighbours to each other. 
		 \label{communication:Ni} 

\end{defn}


Let $\mathcal{F}_i^{(t)}$ \label{communication:Fp_i} denote the set of pruned nodes known by a node $v_i$ at the end of iteration $t$ (with $t\geq 1$). 

\begin{figure}[h]
	\centering
	\includegraphics[width=0.2\textwidth]{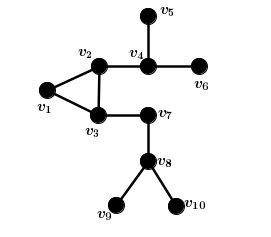}
	\caption[\textit{Illustration of graphs with prunable nodes.}]{\textit{\small{{Illustration of graphs with prunable nodes. }}}}
\label{fig:cycle_node}
\end{figure}



%
%

 Let $\mathcal{N}^\mathrm{up}_{i, t}=\mathcal{N}_{i}\setminus \mathcal{F}^{(t-1)}_{i}$\label{pruning:Niup} denote the set of neighbours of $v_i$ which have not yet been pruned at the beginning of iteration $t$. Functions that a node $v_i$ applies to detect prunable elements  in its neighbourhood after the first iteration are given in Algorithm \ref{algo:pruning} (see functions \textsc{leavesDetection} and \textsc{triangleDetection}). 	
 
 Note that for complete graphs, pruning is not involved because, after the first iteration each node will realise that it its current view of the graph is complete.

So far, we have described how elements of $\mathcal{F}_i^{(1)}$ are identified by node $v_i$. We now consider the case of further pruning 
which is straightforward: new nodes should be pruned when they have no new information to share with their other active neighbours. Thus a node stops relaying neighbouring information to neighbours from which it receives no new information. 

 At the end of iteration $t$, a node $v_i$ considers itself as element of $\mathcal{F}_i^{(t)}$ if it gets all its new information from only  one of its neighbours at that iteration. Also, node $v_i$ prunes $v_j$ if the neighbouring information $\mathcal{N}_j^{(t)}$ sent by $v_j$ to $v_i$ does not contain new information, i.e.
 \begin{equation}
 \label{eq:detection}
 \mathcal{N}_j^{(t)}\subseteq \bigcup_{l<t}\mathcal{N}_i^{(l)}\,.
 \end{equation}
Our hope is that the most central node is among the nodes which are not pruned on termination of the algorithm. Equation \ref{eq:detection} indicates for a node $v_j$ to be pruned, there must be another node $v_i$ which is unpruned because a comparison needs to be done. This is true because there are always unpruned nodes which remain after the first iteration, except a complete graph of at most three nodes in which case pruning is not invoked. This means that at the end of our pruning method, there will always remain some unpruned nodes. 

Nodes in $\mathcal{F}^{(t)}_i$  for $t
\geq 2$ could be viewed as leaves or nodes causing triangles in the subgraph obtained after the removal of all previously pruned nodes. Recall that an unpruned node only interacts with its unpruned neighbours and the number of the unpruned neighbours of a node may get reduced over iterations. So at some iteration an unpruned node can be viewed as leaf if it remains only with one unpruned neighbour. 

 Let 
 $$
 \mathcal{F}_{i, t} = \bigcup_{l=1}^t \mathcal{F}^{(l)}_i\,. 
 $$
  
 The procedure for how a node $v_i$ detects elements of $\mathcal{F}^{(t)}_i$ at the end of each iteration $t\geq 2$ is given in Algorithm \ref{algo:pruning} (see function \textsc{furtherPruningDetection}).

  After describing pruning, we now connect it to eccentricity (see Lemma \ref{lem:eccentricity}) as mentioned above. Let $\mathrm{ecc}_i$\label{pruning:hi} denote the eccentricity  of a node $v_i$, i.e.
 \begin{equation}
 \label{eq:eccentricity}
 \mathrm{ecc}_i =  \max_{v_j\in \mathcal{V}} \delta_{ij}\,.
 \end{equation}
 \begin{lem}
 	\label{lem:eccentricity}
 	If $v_j\in \mathcal{N}_i$ such that 
 	$$
 	v_j\in \bigcup_{t\geq 1}\mathcal{F}_i^{(t)}\,,
 	$$
 	then
 	$$
 	\mathrm{ecc}_j \geq \mathrm{ecc}_i\,.
 	$$
 \end{lem}
 \begin{proof}
 	A node $v_j$ is pruned if it has a direct neighbour $v_i$ from which it can receive new information while at the same time it can not provide new information to that neighbour. Given two direct neighbours $v_j$ and $v_i$ where $v_j$ has been pruned at iteration $t$ by $v_i$, we have Equation \ref{eq:detection}.
 	From the definition of $\mathcal{N}_i^{(t)}$ and Equation \ref{eq:detection}, it is straightforward that $\mathrm{ecc}_j\geq \mathrm{ecc}_i$.
 	%
 	%
 \end{proof}
 
 
 From Lemma \ref{lem:eccentricity} it can be seen that prunable nodes are nodes with relatively high eccentricities. So pruning can not introduce errors when searching for a node of maximum eccentricity centrality. 

\begin{algorithm}
	\begin{algorithmic}[1]
		\scriptsize
		\Procedure{runPruning($\mathcal{N}_i, D$)}{}
		\State \texttt{PruningObject}$\gets$\textsc{pruning}($\mathcal{N}_i, D$)
		\State \texttt{PruningObject}.\textsc{initialOneHop()}
		\State \texttt{PruningObject}.\textsc{initialUpdate()}
		\State \texttt{PruningObject}.\textsc{firstPruningDetection()}
		\While{\NOT \texttt{PruningObject}.\textsc{isEnded()}}
		\State \texttt{PruningObject}.\textsc{nextOneHop()}
		\State \texttt{PruningObject}.\textsc{nextUpdate()}
		\EndWhile
		\EndProcedure
		\State
		\Class{pruning}
		\State \textbf{\underline{Class variables}}
		\State \makebox[1cm][l]{$D$} the pre-set maximum number of iterations
		\State \makebox[1cm][l]{$\mathcal{F}^{(t)}_{i}$} set of pruned nodes known by $v_i$ at the end of iteration $t$
		\State \makebox[1cm][l]{$\mathcal{F}_{i, t}$} set of pruned nodes known by $v_i$ at the end of iteration $t$
		\State \makebox[1cm][l]{$\mathcal{M}_i$} a message queue for messages received by the node
		\State \makebox[1cm][l]{$\mathcal{N}_{i}$} set of immediate neighbours of $v_i$
		\State \makebox[1cm][l]{$\mathcal{N}^\mathrm{up}_{i, t}$} set of neighbours of $v_i$ which are still active up to iteration $t$
		\State \makebox[1cm][l]{$\mathcal{Q}_{ij}$} the one-hop neighbours of node $v_j$ sent to $v_i$ at the first iteration 
		
		\State \makebox[1cm][l]{$\mathcal{N}^{(t)}_{i}$} set of new nodes discovered by $v_i$ at the end of iteration $t$
		\State \makebox[1cm][l]{$\mathcal{N}_{i, t}$} view of communication graph of $v_i$ up to iteration $t$
		\State \makebox[1cm][l]{$c_i$} the closeness centrality of node $v_i$
		\State \makebox[1cm][l]{$\delta_i$} a variable used for computation of $c_i$
		\State \makebox[1cm][l]{$t$} current iteration number
		\State \makebox[1cm][l]{$T$} the actual maximum number of iterations

		\\\hrulefill
		\Constructor{$\mathcal{N}_i, D$}
		\State $(t, D, \delta_i)\gets (0, D, |\mathcal{N}_{i}|)$
		\State $\mathcal{N}^{(0)}_i\gets  \{v_{j}: \forall v_j\in\mathcal{N}_i \}$
		\State $\mathcal{N}_{i, 0}\gets \mathcal{N}^{(0)}_i$.\texttt{clone()} \Comment{``$v_i$ detects its immediate neighbours''}\label{line:initiation_end}
		
		\EndConstructor
						\State
		\Function{initialOneHop()}{}
		
		\For{$v_j \in \mathcal{N}_i$}
		\State $v_i$ sends $\langle \texttt{NeighbouringMessage($i, \mathcal{N}_i^{(0)}$)}\rangle$
		to $v_j$
		\EndFor
		\EndFunction
		\State
				\Procedure{initialUpdate()}{}
		\State $\mathcal{N}^{(1)}_i\gets \emptyset$
		\State $t\gets t + 1$
		\While{$\mathcal{M}_i.\texttt{size()}\geq1$}
		\State $(j,\mathcal{N}_j^{(0)}) \gets \mathcal{M}_i.\texttt{dequeue}()$	\label{line:first_iteration_begin}	
		\State $\mathcal{N}^{(1)}_i\gets \mathcal{N}^{(1)}_i \cup \mathcal{N}^{(0)}_j$
		\Comment{``$v_i$ fuses the messages received''}	
		\State $\mathcal{Q}_{ij}\gets \mathcal{N}^{(0)}_{j}$
		\EndWhile
		\State $\mathcal{N}^{(1)}_i\gets \mathcal{N}^{(1)}_i\setminus \mathcal{N}_{i, 0}$ 
		\State $\mathcal{N}_{i, 1}\gets \mathcal{N}_{i, 0}\cup \mathcal{N}^{(t)}_i$\label{line:first_iteration_end}
		\State $\delta_i\gets \delta_i +t|\mathcal{N}^{(1)}_i|$
		\EndProcedure
				\State
		\Procedure{firstPruningDetection()}{}
		\State \textsc{leavesDetection()} 
		\State  \textsc{triangleDetection()} 
		\State $\mathcal{F}_{i, t}\gets \mathcal{F}^{(t)}_i$.\texttt{clone()}
		\EndProcedure
		\State
		\Procedure{leavesDetection()}{}\Comment{``$v_i$ detects leaves in its neighbourhood}
		\State $\mathcal{F}^{(1)}_i\gets \emptyset$
		\For{$v_j\in \mathcal{N}_i\cup\{v_i\}$}
		\If{$|\mathcal{Q}_{ij}|=1$}

		\algstore{pruningTemplate}
	\end{algorithmic}  
	\caption[\textit{Our proposed pruning method in a failure-free scenario.}]{\textit{Our proposed pruning method in a failure-free scenario. } \textit{  The algorithm gives the code executed for a single node $v_i$. An example of pseudo code is given in the procedure \textsc{runPruning}.}}
	\label{algo:pruning}
\end{algorithm}
\begin{algorithm}
	\begin{algorithmic}[1]
		\algrestore{pruningTemplate}
		\scriptsize
			\State $\mathcal{F}^{(1)}_i\gets \mathcal{F}^{(1)}_i\cup \{v_j\}$
		\EndIf
		\EndFor
		\EndProcedure
		\State
\Procedure{triangleDetection}{}\Comment{``$v_i$ detects elements causing triangles in its neighbourhood}
\For{$v_j\in \mathcal{N}^\mathrm{up}_{i, 1}\cup\{v_i\}$}
\If{  $|\mathcal{Q}_{ij}|=2$}
\State Let $v_f, v_g$ be the two immediate neighbours of $v_j$
\If{$v_f\in\mathcal{Q}_{ig}$}
\State $\mathcal{F}^{(1)}_i\gets \mathcal{F}^{(1)}_i\cup \{v_j\}$
\EndIf

\EndIf
\EndFor

\EndProcedure
\State
\Procedure{nextOneHop()}{}
\If{\NOT \textsc{isEnded()}}
\State $\mathcal{N}^\mathrm{up}_{i, t}\gets \mathcal{N}^\mathrm{up}_{i, t}\setminus \mathcal{F}_{i, t}$

\For{$v_j \in \mathcal{N}^\mathrm{up}_{i, t}$}
\State $v_i$ sends $\langle \texttt{NeighbouringMessage($i, \mathcal{N}_i^{(t-1)}$)}\rangle$
to $v_j$
\EndFor
\EndIf
\EndProcedure
\State
\Procedure{nextUpdate()}{}
\If{\NOT \textsc{isEnded()}}
\State $t\gets t+1$
\State  $\mathcal{N}^{(t)}_i\gets \emptyset$
\While{$\mathcal{M}_i.\texttt{size()}\geq 1$}
\State $(j, \mathcal{N}_j^{(t-1)}) \gets \mathcal{M}_i.\texttt{dequeue}()$
\State $\mathcal{N}^{(t)}_i\gets \mathcal{N}^{(t)}_i \cup \mathcal{N}^{(t-1)}_j$
\Comment{``$v_i$ fuses the messages received''}
\EndWhile
\State $\mathcal{N}^{(t)}_i\gets \mathcal{N}^{(t)}_i\setminus \mathcal{N}_{i, t-1}$ 
\State $\mathcal{N}_{i, t}\gets \mathcal{N}_{ i, t-1}\cup \mathcal{N}^{(t)}_i$
\State \textsc{furtherPruningDetection()} 
\State $\mathcal{F}_{i, t}\gets \mathcal{F}_{i, t} \cup \mathcal{F}^{(t)}_{i}$
\State $\delta_i\gets \delta_i +t|\mathcal{N}^{(t)}_i|$
\Else
\State $T\gets \min(t, D)$
\State $c_i\gets \textsc{closenessCentrality}(T)$
\EndIf

\EndProcedure
		\State
		\Procedure{furtherPruningDetection()}{}\Comment{``$v_i$ detects elements of $\mathcal{F}^{(t)}_i$ in its neighbourhood''}
		\State $\mathcal{F}^{(t)}_i\gets \emptyset$
		\For{$v_j\in \mathcal{N}^\mathrm{up}_{i, t}$}
		\If{$\mathcal{N}^{(t)}_j\subseteq \mathcal{N}_{i, t-1}$}
		\State $\mathcal{F}^{(t)}_i\gets \mathcal{F}^{(t)}_i\cup \{v_j\}$
		\EndIf
		\EndFor
		\If{$|\mathcal{N}^\mathrm{up}_{i, t}|=1$ \AND $\mathcal{N}_i^{(t)}\neq \emptyset$ }
		\State $\mathcal{F}^{(t)}_i\gets \mathcal{F}^{(t)}_i\cup \{v_i\}$
		\EndIf
		\EndProcedure
		\State

		\Function{closenessCentrality}{$T$}
		\If{$v_i \in \mathcal{F}_{i, T}$}
		\State \Return $ 0$
		\EndIf
		\State \Return $ \frac{|\mathcal{N}_{i, T}|-1}{\delta_i}$
		\EndFunction
		
		\State
		\Function{isEnded()}{}
		\State \Return $t=D$\OR $\mathcal{N}_i^{(t)}=\emptyset$ \OR $v_i \in \mathcal{F}_{i, t}$
		\EndFunction
		\EndClass
	\end{algorithmic}  
\end{algorithm}

\begin{example}
	Consider execution of Algorithm \ref{algo:pruning} on the communication graph in Figure \ref{fig:cycle_node}, with $D=4$. The results of our pruning method on this graph are presented in Table \ref{tab:pruning_example} and discussed below.
	
	\begin{table}[h]
		\centering
		\scalebox{0.7}{
			\begin{tabular}{|c|c|c|c|c|c|}
				\hline
				\textbf{Node sets} & $\mathrm{ecc}_i$ & $t=1$ & $t=2$ & $t=3$ & $t=4$ \\
				\hline
				$\mathcal{F}_{1}^{(t)}$ & $4$ &$v_1$&$\perp$&$\perp$&$\perp$\\
				\hline
				$\mathcal{F}_{2}^{(t)}$& $4$ &$v_1$&$v_4$&$v_2$&$\perp$\\
				\hline
				$\mathcal{F}_{3}^{(t)}$& $3$ &$v_1$&$\emptyset$&$v_2, v_7$&$\top$\\
				\hline
				$\mathcal{F}_{4}^{(t)}$ & $5$ &$v_5, v_6$&$v_4$&$\perp$&$\perp$\\
				\hline
				$\mathcal{F}_{5}^{(t)}$ & $6$ &$v_5$&$\perp$&$\perp$&$\perp$\\
				\hline
				$\mathcal{F}_{6}^{(t)}$ & $6$ &$v_6$&$\perp$&$\perp$&$\perp$\\
				\hline
				$\mathcal{F}_{7}^{(t)}$ & $4$ &$\emptyset$&$v_8$&$v_7$&$\perp$\\
				\hline
				$\mathcal{F}_{8}^{(t)}$ & $5$ &$v_9, v_{10}$&$v_8$&$\perp$&$\perp$\\
				\hline
				$\mathcal{F}_{9}^{(t)}$ & $6$ &$v_9$&$\perp$&$\perp$&$\perp$\\
				\hline
				$\mathcal{F}_{10}^{(t)}$ & $6$ &$v_{10}$&$\perp$&$\perp$&$\perp$\\
				\hline
		\end{tabular}}
		\caption[\textit{Table indicating pruned nodes, identified at each node after each iteration using the graph in Figure \ref{fig:cycle_node}.}]{\textit{Table indicating  pruned nodes, identified at each node after each iteration using the graph in Figure \ref{fig:cycle_node}. $\perp$ indicates that the corresponding node has pruned itself and $\top$ indicates that the node has reached an equilibrium. }}
		\label{tab:pruning_example}
		\vspace{-0.3cm}
	\end{table}
After the first iteration, each node needs to identify prunable nodes in its neighbourhood, including itself. Using Algorithm \ref{algo:pruning}, $\mathcal{F}_1^{(1)}=\{v_1\}$, $v_1$ will prune itself. Also, $\mathcal{F}_2^{(1)}=\mathcal{F}_3^{(1)}=\{v_1\}$, $v_2$ and $v_3$ will also prune $v_1$.  At the first iteration, all the leaves (i.e. $v_5, v_6, v_9$ and $v_{10}$) are pruned; they all have $\mathrm{ecc}_i=6$. But, a non-leaf node, $v_1$ (with $\mathrm{ecc}_i=4$), is also pruned on the first iteration. 

At the beginning of iteration $t=2$, $v_1, v_5, v_6, v_9$ and $v_{10}$ are no longer involved since they have been identified as prunable nodes at the end of the previous iteration; $\mathcal{F}_2^{(2)}=\mathcal{F}_4^{(2)}=\{v_4\}$ and $\mathcal{F}_7^{(2)}=\mathcal{F}_8^{(2)}=\{v_8\}$; and the node $ v_3$ does not identify any prunable node after this iteration. At this iteration, the remaining nodes with high $\mathrm{ecc}_i$ (i.e. nodes $v_4$ and $v_8$) are pruned. The results for the remaining iterations are shown in Table \ref{tab:pruning_example}.
\end{example}

In our proposed distributed system, at the end of each iteration each node is aware of whether each of its direct neighbour is pruned or not). A pruned node neither sends a message nor waits for a message. So when a node is still unpruned, it knows which immediate neighbours to send messages to and which to wait for messages from. This prevents the nodes from suffering from starvation or deadlock in failure-free scenarios \cite{coulouris2005distributed}.
\subsection{Communication analysis}
\label{sec:analysis}
In this section, we evaluate the impact of pruning on the communication requirements for view construction. 
%
Let  $u_i^{(t)}$ denote the number of neighbours of $v_i$  which have been pruned at the end of iteration $t$.
Let $Y^{(D)}_i$ \label{communication:Yi} and $P^{(D)}_i$ \label{communication:Pi} be the number of messages that the node $v_i$ receives according to Algorithm \ref{algo:construction0} and the number of messages $v_i$ receives through the use of pruning in Algorithm \ref{algo:pruning} for $D$ rounds respectively. 
We expect the number of messages  any node $v_i$ saves due to pruning to satisfy
\begin{equation}
\label{eq:acualmes}
\Delta^{(D)}_{i}\equiv Y^{(D)}_i-P^{(D)}_i\geq 0\,.\end{equation}



A node $v_i$ receives $d_i$ messages at the end of each iteration using the YTQ method. Recall that our proposed pruning and the YTQ methods can also terminate when an equilibrium is reached. Let $H_i$ denote the iteration after which a node $v_i$ applying the YTQ and our pruning methods reaches an equilibrium. This value is the same for both algorithms because at iteration $t$, a node $v_i$ (which should be an unpruned node using our proposed method) has the same view using both algorithms. Note that for our pruning method, a pruned node does not reach an equilibrium and it is not possible to prune all nodes before equilibrium. Let $h_i^{(t)}$ be the number of neighbours of $v_i$ which have reached equilibrium at the end of iteration $t$. If  at least one neighbour of an unpruned node $v_i$ has reached equilibrium by iteration $t$, then $v_i$ will reach equilibrium by iteration $t+1$.
For the YTQ method, 
\begin{equation}
\label{eq:you_all_free}
Y^{(D)}_i =  \sum_{t=1}^{\min(D, H_i)} \left(d_i-\sum_{l=0}^{t-1} h_i^{(l)}\right)\,.
\end{equation}

Let $L_i$ denote the round at which a node $v_i$ is pruned ($L_i=+\infty$ for unpruned node $v_i$).
\begin{lem}
	\label{lem:pruning_free} 
	 The number of messages received by a node $v_i\in \mathcal{V}$ in Algorithm \ref{algo:pruning} is	 
	$$P^{(D)}_i = \sum_{t=1}^{\min(D, H_i, L_i)} \left(d_i-\sum_{l=0}^{t-1} \left(h_i^{(l)}+u^{(l)}_i\right)\right)\,.$$
\end{lem}
\begin{proof}
	At the end of each iteration $t$, the node
	$v_i$ receives $\left(d_i-\sum_{l=0}^{t-1} \left(h_i^{(l)}+u^{(l)}_i\right)\right)$ messages. Note that $u_i^{(0)}=h_i^{(0)}=0$. Also a node can stop
	interacting with other nodes after it is pruned. If the node $v_i$ is pruned at the end of iteration $\min(D, H_i, L_i)$, then it stops receiving messages. 
\end{proof}

\begin{thm}
	\label{thm:delta_free} 
	 The number of messages saved by a node $v_i\in \mathcal{V}$ in a failure-free scenario is	 
	\begin{align*}
	\Delta^{(D)}_{i} &= \sum_{t=1}^{\min(D, H_i, L_i)} \sum_{l=0}^{t-1}  u^{(l)}_i+ \sum_{t=\min(D, H_i, L_i)+1}^{\min(D, H_i)} \left(d_i-\sum_{l=0}^{t-1} h_i^{(l)}\right)\,.
	\end{align*}
\end{thm}
\begin{proof}
	This is straightforward by Equations \ref{eq:acualmes} and \ref{eq:you_all_free}, and Lemma \ref{lem:pruning_free}.
\end{proof}

It is clear that pruned nodes build very limited views  as they stop interacting with others once they are pruned. These nodes would have built broader views using the YTQ method \cite{you2017distributed}. Thus if considering applications where all nodes are required to build broader views, our pruning method is not recommended.

\subsection{Communication failure}
We also extend our pruning method to take into account communication failures during view construction. For failure management, we simply incorporate the neighbour coordination approach proposed by Sheth et al. \cite{sheth2005decentralized} into our pruning method. We found the coordination approach for failure management most suitable for our pruning method because each node can monitor the behaviour of its immediate neighbours and report failures and recoveries if detected, which suits our decentralised approach well. Details of the extended version of pruning with communication failure can be found in \cite{jmf2020}. (It should be noted that we exclude details of failure management so that we can focus on the main contribution of this work).

\section{Experimental investigation}
\label{sec:results}
For the comparison of our proposed method with the YTQ method \cite{you2017distributed} in terms of the number of messages, we consider the total and the maximum number of messages received per node. We wish to see the impact of our pruning method  on message complexity in the entire network. 
We wish to reduce the maximum number of messages per node because if the communication time per message is the bottleneck, reducing only the total number of messages may not be helpful.

Our experiments considered the following cases, in an attempt to comprehensively test the proposed
approach:
\begin{enumerate}[(1)]
	\compresslist
	\item Comparison of the number of messages received by nodes for each method on various networks. We also used a Wilcoxon signed-rank \cite{wilcoxon1992individual} test and the effect size \cite{cohen1962statistical} to verify whether the mean differences of the number of messages between pruning and the YTQ method are significantly different.
	\item Comparison of the approximated most central nodes obtained with our method to the YTQ method. A good approximation should choose a most central node with a small distance to the exact most central node. We also used a Wilcoxon signed-rank test and the effect size to verify whether the mean differences of shortest path distances between approximate central nodes obtained using the YTQ and our pruning methods with respect to the exact most central node on some random graphs are significantly different.
	
	In Section \ref{sec:description}, we showed that pruning is related to node eccentricity which allows us to ensure approximation of closeness centrality using pruning because eccentricity and closeness centralities are positively and strongly correlated for various types of graphs \cite{batool2014towards, meghanathan2015correlation}. We run simulation experiments to determine the Spearman's $\rho$ \cite{spearman1961general} and Kendall's $\tau$  \cite{kendall1938new} coefficients between eccentricity and closeness centralities.
	We consider the Spearman's $\rho$ and Kendall's $\tau$ coefficients because they are appropriate correlation coefficients to measure the correspondence between two rankings. 

\end{enumerate}
For the hypothesis tests, the significance level we use is $0.01$.
\subsection{Experimental setup}
We implemented  our pruning method
using Python and \texttt{NetworkX} \cite{hagberg2013networkx}. 
Our simulation was run on  two HPC (High Performance Computing) clusters hosted by Stellenbosch University. 
Our code can be found at \url{https://bitbucket.org/jmf-mas/codes/src/master/network}. 


We ran several simulations with random graphs (generated as discussed below), as well as some real-world networks. 
\subsubsection{Randomly generated networks}
We used a $200$x$200$ grid with integer coordinates and generated $50$ random connected undirected graphs as follows:  
We generated $N$ uniformly distributed grid locations (sampling without replacement) as nodes. The number of nodes, $N$, was sampled uniformly from $[50, 500]$. Two nodes were connected by an edge if the Euclidean distance between them was less than a specified communication range $d=8$. 
 The number of edges and the diameter for these graphs were in the intervals $[50, 2000]$ and $[20, 60]$ respectively\textemdash see Figure \ref{fig:random_property}. 
\begin{figure}[h]
	\centering
		\includegraphics[width=0.24\textwidth]{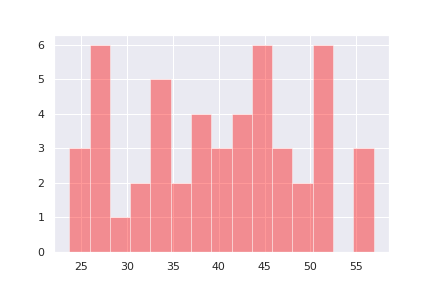}
	\caption[\textit{Histogram of diameters of random graphs.}]{\textit{\small{{Diameters of random graphs. We use a binwidth of $2$. 
	}}}}
	\label{fig:random_property}
	\vspace{-0.3cm}
\end{figure}

%
\subsubsection{Real-world networks}
The $34$ real-world graphs we consider are a phenomenology collaboration network \cite{leskovec2007graph}, a snapshot of the Gnutella peer-to-peer network \cite{leskovec2007graph}, and $32$ \textit{autonomous graphs} \cite{leskovec2005graphs}. The phenomenology collaboration network represents research collaborations between authors of scientific articles submitted to the Journal of High Energy Physics. In the Gnutella peer-to-peer network, nodes represent hosts in the Gnutella network and edges represent connections between the Gnutella hosts. Autonomous graphs are graphs composed of links between Internet routers. These graphs represent communication networks based on Border Gateway Protocol logs. Some characteristics of some of these networks are given in Table \ref{tab:autonomous_graph}. 

\subsection{Results and discussion}
\label{sec:network_results} 
\subsubsection{Average and maximum number of messages}
\begin{figure}[h]
	\centering	
	\begin{subfigure}[h*]{0.23\textwidth}
		\includegraphics[width=\textwidth]{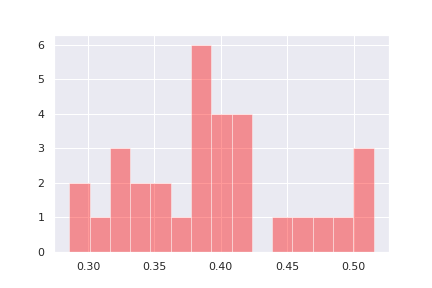}
		\caption{}
		\label{fig:avg_msg_auto_proportion}
	\end{subfigure}
	\begin{subfigure}[h*]{0.23\textwidth}
	\includegraphics[width=\textwidth]{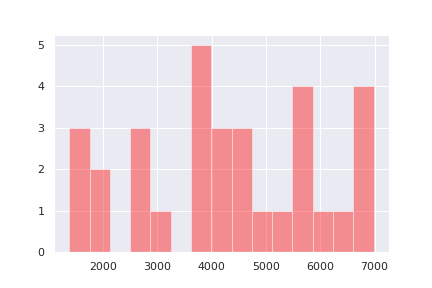}
	\caption{}
	\label{fig:max_msg_auto_difference}
\end{subfigure}
	\caption[\textit{Histograms showing the difference in the average and maximum number of messages on the $32$ autonomous graphs.}]{\textit{\small{{Differences in the number of messages between the YTQ method and our pruning method on the $32$ autonomous graphs. Positive values indicate that our pruning method outperfoms the YTQ method.  \textbf{(\ref{fig:avg_msg_auto_proportion})}: Reduction in the average number of messages.
	\textbf{(\ref{fig:max_msg_auto_difference})}: Reduction in the maximum number of messages.  }}}}
	\label{fig:avg_msg_auto}
	\vspace{-0.3cm}
\end{figure}

\begin{figure}[h]
	\centering
	\includegraphics[width=0.24\textwidth]{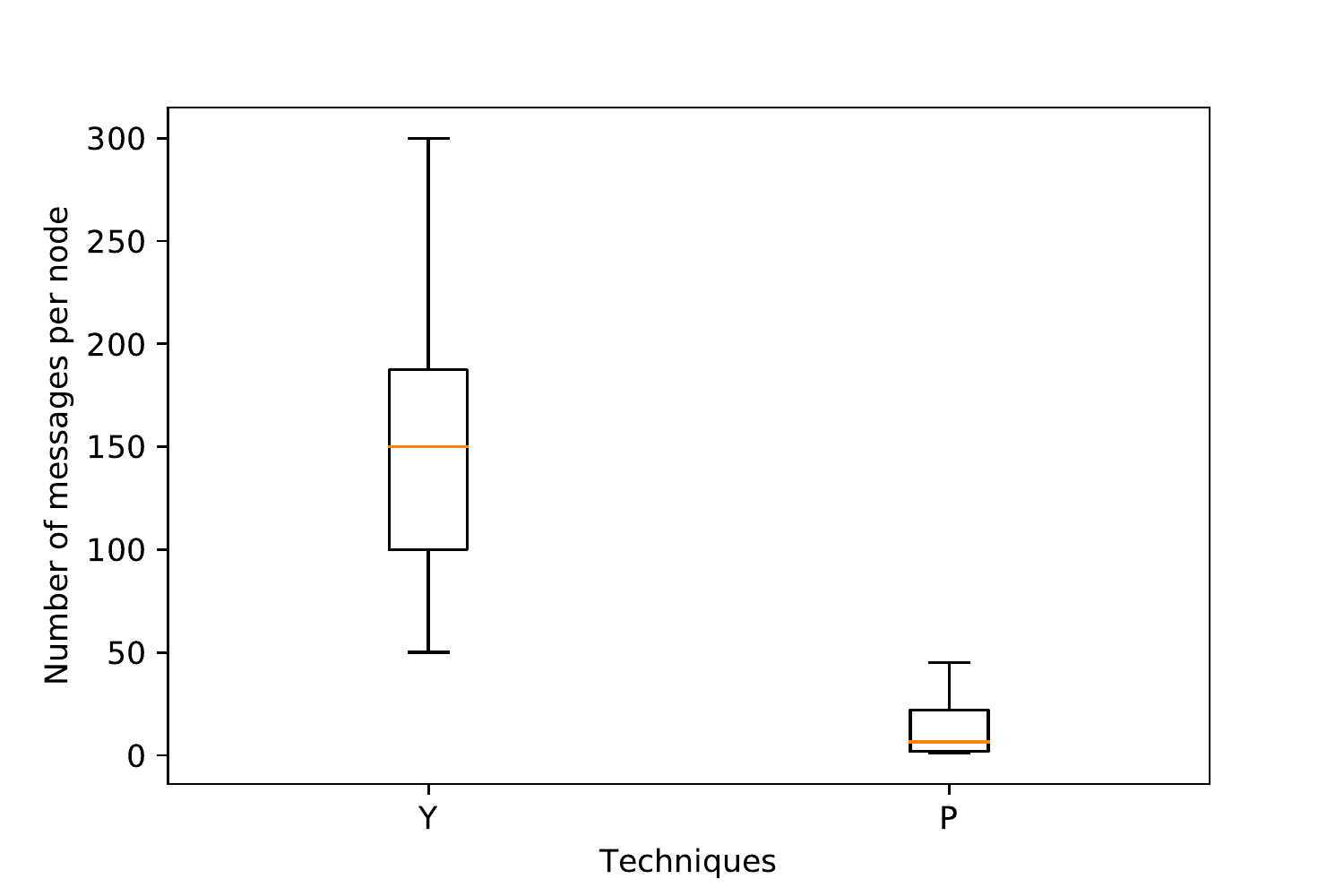}
	\caption[\textit{Number of messages using YTQ and pruning methods on one random graph.}]{\textit{\small{{Numbers of messages per node using each method on a random graph with $D=10$. On the horizontal axis Y and P denote the YTQ and pruning methods respectively.}}}}
	\label{fig:approximate}
	\vspace{-0.3cm}
\end{figure}
Our experiments illustrate the improved communication performance over the YTQ method in \cite{you2017distributed} resulting from pruning. Figure \ref{fig:avg_msg_auto} shows differences in the averages and in the maximum number of messages per node between the YTQ and our pruning methods on the $32$ autonomous graphs. 
We see that the approach reduced communication by $30-50\%$ on average for all network, with over $75\%$ of networks reducing their maximum number of messages by $30\%$ or more.

\begin{table}[h]
	\centering
	\scriptsize
		\scalebox{0.7}{\begin{tabular}{ |l|l|l||l|l||l|l| }
		\hline
		\textbf{Nodes}&\textbf{Edges}&\textbf{Diameter}&\textbf{Y}  & \textbf{P}&\textbf{Ymax}  & \textbf{Pmax}\\
		\hline
		\hline
		\multicolumn{7}{|l|}{\textbf{Three autonomous networks}}\\
		\hline
		$1486$&$3422$&$9$&\textit{28.0}&$ 8.8$&\textit{7005}&$ 1413$\\
		\hline
		$2092$&$4653$&$9$&\textit{27.2}&$ 7.9$&\textit{3312}&$ 552$\\
		\hline
		$6232$&$13460$&$9$&\textit{25.4}&$ 6.9$&\textit{7295}&$ 1459$\\
		\hline
		\multicolumn{7}{|l|}{\textbf{Phenomenology collaboration network}}\\
		\hline
		$10876$&$39994$&$9$&\textit{52.2}& $ 14.3$&\textit{721}&$ 120$\\
		\hline
		\multicolumn{7}{|l|}{\textbf{Gnutella peer-to-peer network}}\\
		\hline
		$9877$&$25998$&$13$&\textit{63.0}&$ 8.3$&\textit{3705}&$ 787$\\
		\hline
	\end{tabular}}
	\caption[\textit{Results for a sample of $5$ real-world networks ($1$ phenomenology collaboration network, $1$ Gnutella peer-to-peer network and $3$ autonomous networks) on the YTQ method and our proposed pruning method.}]{\textit{Graph properties and the number of messages with each approach for five real-world networks (the phenomenology collaboration network, the Gnutella peer-to-peer network, and three autonomous networks). Y(Ymax) and P(Pmax) denote average(maximum) number of messages for the YTQ and pruning methods respectively. 
		}}
	\label{tab:autonomous_graph}
\end{table}

Table \ref{tab:autonomous_graph} shows the average  and the maximum number of messages per node for five real-world networks respectively.  
The number of messages per node for each technique on one random network are contrasted in Figure \ref{fig:approximate}. 
%
The results confirm that the  pruning method is better than the YTQ method (Figure \ref{fig:avg_msg_auto}).  

\subsubsection*{Hypothesis test}

We observed a $p$-value of $7.96\times 10^{-90}$ and an effect size of $21.1140$  between the number of messages obtained with our pruning and the YTQ methods on $50$ random graphs containing $500$ nodes each. The $p$-value is less than the threshold $0.01$, so the means of the number of messages using the YTQ method against pruning are significantly different. In terms of effect size, according to the classification in Gail and Richard \cite{effectsize}, the effect size between our pruning and the YTQ methods is large ($e\geq 0.8$). So the means of the number of messages using the YTQ method against pruning differ markedly.


We also observed a reduction in total running time  and memory usage from our approach. So no adverse effect on power usage from the approach.

\subsubsection{Quality of selected most central node}
First, for various graphs considered here, the averages and standard deviations of the Spearman's $\rho$ and Kendall's $\tau$ coefficients between eccentricity and closeness centralities are $0.9237\pm0.0508$ and $0.7839\pm0.0728$ respectively, and the correlation coefficients were all positive. This shows that there is predominantly a fairly strong level of correlation between eccentricity and closeness centralities for various graphs considered here. This confirms the results by Batool and Niazi \cite{batool2014towards}, and Meghanathan \cite{meghanathan2015correlation}. 
Note that we chose real-world graphs and random graphs modelled on what might be realistic for a sensor network\textemdash so no attempt to choose graphs from models with high correlation.

Tables \ref{tab:shortest_path_1} shows the shortest path distances between the exact most central node and approximated most central nodes using our pruning method and the YTQ method \cite{you2017distributed} for two random graphs. 
Figure \ref{fig:closeness} shows differences of shortest path distances between the exact most central node and approximated central nodes obtained with the YTQ and our pruning methods on two random graphs. 
For each of the graphs, we randomly vary $D$ in a range of values smaller than the diameter of the graph. 

\begin{table}[h]
	\centering
	\scriptsize
	\scalebox{0.9}{	\begin{tabular}{ |l|c|c||c|c| }
			\hline
\textbf{$D$}&\textbf{Y$_1$}&\textbf{P$_1$}&\textbf{Y$_2$}&\textbf{P$_2$}\\
\hline
$2$ &5&5&14&14\\
\hline
$6$ &$14$&10&$19$&5\\
\hline
$10$ &$13$&1&$20$&0\\
\hline
$14$ &$14$&4&$19$&2\\
\hline
$18$ &$8$&0&0&0\\
\hline
$22$ &0&0&0&$2$\\
\hline
$26$ &0&0&0&0\\
\hline
	\end{tabular}}
	\caption[\textit{Shortest path distances between the exact most central node and approximated most central node for failure-free situations.}]{\textit{Shortest path distances between the exact most central node and approximated most central node using pruning and the YTQ method. We use two random graphs, one with $70$ nodes and diameter of $35$, and another with $72$ nodes and diameter of $32$.  One method achieves better approximations than another if the distance of the selected node from the true most central node is smaller. Y$_i$ and P$_i$ indicate shortest path distances for the YTQ and pruning methods on the $i$-th random graph respectively.}}
	\label{tab:shortest_path_1}
\end{table}

\begin{figure}[h]
	\centering
	\begin{subfigure}[h*]{0.23\textwidth}
		\includegraphics[width=\textwidth]{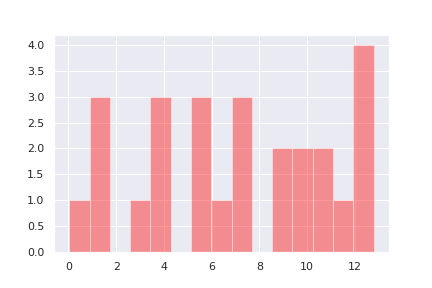}
		\caption{}
		\label{fig:closeness_0}
	\end{subfigure}
	\begin{subfigure}[h*]{0.23\textwidth}
		\includegraphics[width=\textwidth]{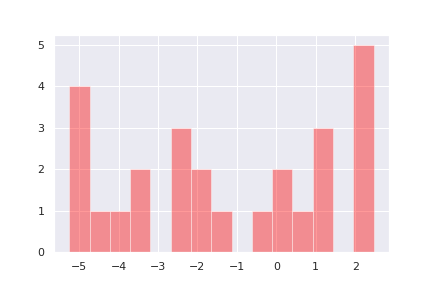}
		\caption{}
		\label{fig:closeness_1}
	\end{subfigure}
	\caption[\textit{Histograms of differences of shortest path distances between approximated central nodes obtained using the YTQ and our pruning methods with respect to the exact most central node on some random graphs for failure-free cases.}]{\textit{\small{{Differences of shortest path distances between approximated central nodes obtained using the YTQ and our pruning methods with respect to the exact most central node on two random graphs: \textbf{(\ref{fig:closeness_0})}: the first (diameter of $26$, $125$ nodes and $180$ edges) and \textbf{(\ref{fig:closeness_1})}: second (diameter of $36$, $289$ nodes and $597$ edges)random graphs. Positive values indicate that our method is better than the YTQ method.}}}}
	\label{fig:closeness}
	\vspace{-0.3cm}
\end{figure}
The evaluation of node centrality based on a limited view of the communication graph has an impact on the choice of the most central node.
When using our pruning and the YTQ methods to choose a leader based on closeness centrality, the methods can yield different results under the same conditions. We found that (Table \ref{tab:shortest_path_1} and Figure \ref{fig:closeness}) our pruning method generally gave better approximations to closeness centrality than the YTQ method when $D$ is considerably smaller than the diameter, with the results of the YTQ method improving as $D$ increases. This supports our claim that our pruning method effectively identifies nodes which should not be chosen as leaders  as they are highly unlikely to have the highest closeness centrality. 
Note that, even though the two methods sometimes give the exact most central nodes for some $D$ (for example for $D=26$ in Table \ref{tab:shortest_path_1}), these exact most central nodes are not guaranteed. 


The reason why the YTQ method yields poor results when $D$ is smaller than the diameter of the graph is as follows. When some of the nodes have different views of the communication graph and each evaluates its closeness centrality based only on its own view, a node with small but unknown exact closeness centrality may have a high estimated closeness centrality. This can lead to poor conclusions. 
In the YTQ method, the central node is selected from all nodes.
The advantage of the pruning method is that only unpruned nodes compute their approximate closeness centralities, i.e. the many nodes that are pruned are no longer candidates for central nodes.  This reduces the chance of yielding poor performance as the central node is selected from a shorter list of candidates, i.e. the unpruned nodes.

\subsubsection*{Hypothesis test}

We observed a $p$-value of $0.1197$ between the results obtained with our pruning and the YTQ methods on $50$ random graphs of $500$ nodes each. The $p$-value is greater than the threshold $0.01$, so there is no significant difference between the means for the two approaches.  This means that the qualities of the selected most central nodes using both methods are almost the same. This is beneficial to pruning\textemdash though they both provide almost the same qualities of selected most central nodes, pruning reduces the number of messages significantly compared to the YTQ method \cite{you2017distributed}.
\section{Conclusion}
\label{sec:conclusion}
We proposed an enhancement to a benchmark method \cite{you2017distributed} for view construction. The main motivation of this enhancement was to reduce the amount of communication: we aim to reduce the number of messages exchanged between nodes during interaction. Given a network, some nodes can be identified early as being unlikely to be central nodes. Our main contribution was noting that we can identify such nodes and reduce communication by pruning them. 

Our proposed method improves the benchmark method in terms of number of messages. 
We found that reduction of the number of messages has a positive impact on running time and memory usage \cite{jmf2020}.

\textbf{Future work.}
	Our  message counting model ignores the fact that in large networks, messages comprise multiple packets. Analysis of the savings of our approach in terms of the actual amount of data communicated could be investigated in future. Further, it may be possible to identify further types of prunable nodes and consider richer classes of graphs for assessment.

\balance
\bibliography{references}

\end{document}